\documentclass[12pt,a4paper,reqno,intlimits,sumlimits]{amsart}

\usepackage{amssymb}
\usepackage[T1]{fontenc}
\usepackage[numbers, sort]{natbib}
\bibpunct{[}{]}{,}{n}{,}{,}

\usepackage[english]{babel}
\usepackage{color}
\usepackage[mathscr]{euscript}
\usepackage{xspace}
\usepackage{enumerate}
\usepackage{epsfig,rotating}
\input{xy}\xyoption{all}
\usepackage{todonotes}

\DeclareMathAlphabet{\mathpzc}{OT1}{pzc}{m}{it}

\usepackage{geometry}
\begin{document}

\theoremstyle{plain}
\newtheorem{theorem}{Theorem}[section]
\newtheorem{lemma}[theorem]{Lemma}
\newtheorem{proposition}[theorem]{Proposition}
\newtheorem{claim}[theorem]{Claim}
\newtheorem{corollary}[theorem]{Corollary}
\newtheorem{axiom}{Axiom}

\theoremstyle{definition}
\newtheorem{remark}[theorem]{Remark}
\newtheorem{remarks}[theorem]{Remarks}
\newtheorem{note}{Note}[section]
\newtheorem{definition}[theorem]{Definition}
\newtheorem{example}[theorem]{Example}
\newtheorem{condition}{\bf Condition}
\newtheorem{assumption}{\bf Assumption}
\renewcommand*{\thecondition}{\Roman{condition}}
\renewcommand*{\theassumption}{\Roman{assumption}}
\renewcommand{\theequation}{\thesection.\arabic{equation}}
\numberwithin{equation}{section}

\let\SETMINUS\setminus
\renewcommand{\setminus}{\backslash}

\def\stackrelboth#1#2#3{\mathrel{\mathop{#2}\limits^{#1}_{#3}}}


\def\F{\ensuremath{\mathcal{F}}}
\def\bA{\mathbb{A}}
\def\bD{\mathbb{D}}
\def\bF{\mathbb{F}}
\def\bG{\mathbb{G}}
\def\bH{\mathbb{H}}
\def\R{\ensuremath{\mathbb{R}}}
\def\N{\ensuremath{\mathbb{N}}}
\def\Rp{\mathbb{R}_{\geqslant0}}
\def\Rm{\mathbb{R}_{\leqslant 0}}
\def\C{\ensuremath{\mathbb{C}}}
\def\U{\ensuremath{\mathcal{U}}}
\def\I{\mathcal{I}}

\def\P{\ensuremath{\mathrm{I\kern-.2em P}}}
\def\Q{\mathbb{Q}}
\def\E{\ensuremath{\mathrm{I\kern-.2em E}}}

\def\e{\mathrm{e}}
\def\ud{\ensuremath{\mathrm{d}}}
\def\dt{\ud t}
\def\ds{\ud s}
\def\dx{\ud x}
\def\dy{\ud y}
\def\du{\ud u}
\def\dv{\ud v}
\def\dsdx{\ensuremath{(\ud s, \ud x)}}
\def\dtdx{\ensuremath{(\ud t, \ud x)}}

\newcommand{\cA}{{\mathcal{A}}}
\newcommand{\cB}{{\mathcal{B}}}
\newcommand{\cD}{{\mathcal{D}}}
\newcommand{\cE}{{\mathcal{E}}}
\newcommand{\cF}{{\mathcal{F}}}
\newcommand{\cG}{{\mathcal{G}}}
\newcommand{\cH}{{\mathcal{H}}}
\newcommand{\cK}{{\mathcal{K}}}
\newcommand{\cM}{{\mathcal{M}}}
\newcommand{\cN}{{\mathcal{N}}}
\newcommand{\cO}{{\mathcal{O}}}
\newcommand{\cP}{{\mathcal{P}}}
\newcommand{\cT}{{\mathcal{T}}}
\newcommand{\ha}{{\mathbb{H}}}
\newcommand{\Mloc}{\mathcal{M}_{\text{loc}}}
\newcommand{\tMloc}{\mathcal{M}_{\text{\emph{loc}}}}
\newcommand{\indik}{{\mathbf{1}}}
\newcommand{\ifdefault}[1]{\ensuremath{\mathbf{1}_{\{#1  > \tau\}}}}
\newcommand{\ifnodefault}[1]{\ensuremath{\mathbf{1}_{\{#1 \leq \tau\}}}}
\newcommand{\be}{\begin{equation}}
\newcommand{\ee}{\end{equation}}
\newcommand{\lsi}{\left[\negthinspace\left[}
\newcommand{\rsi}{\right]\negthinspace\right]}
\newcommand{\lsir}{\left(\negthinspace\left(}
\newcommand{\rsir}{\right)\negthinspace\right)}
\newcommand{\mutilde}{\tilde{\mu}}
\newcommand{\Wtilde}{\widetilde{W}}
\newcommand{\What}{\widehat{Y}}
\newcommand{\Zhat}{\widehat{Z}}
\newcommand{\citeN}{\citet}

\title[Asset price models with a change point]{Information, no-arbitrage and completeness \\ for asset price models with a change point}

\author{Claudio Fontana}

\address{C. Fontana, M. Jeanblanc - Laboratoire Analyse et Probabilité, Universit\'e d'\'Evry Val d'Essonne, 23 boulevard de France, F-91037 \'Evry Cedex, France}
\email{claudio.fontana@univ-evry.fr, monique.jeanblanc@univ-evry.fr}

\author{Zorana Grbac}

\address{Z. Grbac - Laboratoire de Probabilit\'es et Mod\`eles Al\'eatoires, Universit\'e Paris Diderot, Case 7012, 
75205 Paris Cedex 13, France}
\email{grbac@math.univ-paris-diderot.fr}

\author{Monique Jeanblanc}

\author{Qinghua Li}

\address{Q. Li - Institute of Mathematics, Humboldt University Berlin, Unter den Linden 6, 10099 Berlin}
\email{ms.qinghuali@gmail.com}

\thanks{This research benefited from the support of the \emph{Chaire Risque de Crédit}, Fédération Bancaire Française. The research of C.F. was supported by a Marie Curie Intra European Fellowship within the 7th European Community Framework Programme. Z.G. gratefully acknowledges the financial support from the DFG Research Center MATHEON. Useful comments by two anonymous referees are gratefully acknowledged}

\begin{abstract}
We consider a general class of continuous asset price models where the drift and the volatility functions, as well as the driving Brownian motions, change at a random time $\tau$.
Under minimal assumptions on the random time and on the driving Brownian motions, we study the behavior of the model in all the filtrations which naturally arise in this setting, establishing martingale representation results and characterizing the validity of the NA1 and NFLVR no-arbitrage conditions.
\end{abstract}

\date{\today}\maketitle\pagestyle{myheadings}\frenchspacing

\section{Introduction}

    The behavior of financial asset prices is often subject to certain random events that result in abrupt changes in their dynamics. The random time when such an event occurs is called a change point. For example, a sudden adjustment in the interest rates, a default of a major financial institution, or the release of some political news could all have an impact on the asset price. Although these events are not caused by the price evolution of the individual asset, their occurrence may change the asset price dynamics. In this case the change point is said to be \emph{exogenous} to the model.  On the other side, events linked to the assets themselves can also cause a change in their dynamics, e.g. an asset price crossing a certain threshold. Such a change point is said to be \emph{endogenous} to the model.

    In this paper we study a general class of models with a change point, which are able to capture both of the above-mentioned situations. The dynamics of the asset price is modeled as a stochastic exponential of a  continuous semimartingale, whose
    drift, volatility and driving Brownian motions change at the random time. Only minimal assumptions on the non-negative random variable
    mathematically representing the change point are imposed, in the sense that if such assumptions are violated, then the semimartingale property of the driving Brownian motions may be lost when changing filtrations and pathological forms of arbitrage are possible. The drift and the volatility are stochastic and depend on time and on the current state of the process and the two Brownian motions are not necessarily mutually independent. In particular, the random time is allowed to depend on the Brownian motions.

     We aim at understanding the structure and the behavior of this class of models in all the filtrations which naturally arise in this setting. These filtrations represent different levels of information, starting from the minimal knowledge of only the asset price process up to some time $t$ to the full knowledge of the driving Brownian motions up to time $t$ together with the knowledge of the random time already at time $t=0$. The different filtrations are obtained either as progressive or initial enlargements of a reference filtration with respect to the random time.
Using enlargement of filtration techniques, we study the no-arbitrage properties of the model in each filtration. In particular, we characterize the absence of arbitrages of the first kind, which is equivalent to a square-integrability property of the market price of risk process. In turn, the latter condition allows us to obtain martingale representation theorems. Combining these results we give for all filtrations a complete characterization of all equivalent local martingale measures and of market completeness.
Even though we are not specifically concerned with the detection of the change point, we find that in the
     case when the two volatility regimes are distinct, the random time is actually a stopping time with respect to the price process filtration. On the other side, if the two volatility regimes coincide everywhere, then the change point is not observable.

 Our study of this type of models was inspired by a recent paper by \citeN{CawstonVostrikova11}, where an exponential asset price model driven by two independent L\'evy processes and with an independent change point is developed and analyzed in the context of utility maximization. In comparison with that paper, we refrain from imposing any independence assumption, giving a complete description of the model in all possible filtrations in a Brownian setting. 
We refer to \citeN{CawstonVostrikova11} for a comprehensive literature review on change point models (which have been extensively studied by A. N. Shiryaev and co-authors, see e.g. \citeN{Shiryaev09}), as well as their applications. The focus in these papers is often on the problem of (quickest) detection of the change point.

  The model presented in this paper can also be seen as a regime switching model. By regime switching models we mean models in which the drift and the volatility of the price
  process are functions of a process taking finitely many values, which are interpreted for example as states of the economy. This underlying state process is usually assumed to be a Markov chain in order to ensure the analytical tractability of the model. A special case of our model, in which the two Brownian motions are assumed to be the same and the random time is the first jump time to an absorbing state 1 of a Markov process with two states $\{0, 1\}$, is interpreted as a regime switching model according to the above definition.
Regime switching models have been widely employed in statistics and financial modeling, see for instance the volume \citeN{ElliottMamon07}, the survey paper \citeN{Guo04} and, for econometric applications, Chapter 22 of \citeN{Hamilton94} and the references therein.

The paper is organized as follows. In Section 2 we introduce the notation and the general setting of the model. Section 3 proves the well-posedness of the main SDE defining the model and studies the properties of the model in two progressively enlarged filtrations. 
In Section 4 we analyze the model in its own filtration.  
In particular, we study two special cases where the volatility functions differ or coincide everywhere, respectively.  
Section 5 is dedicated to the study of the properties of the model in two initially enlarged filtrations. 
Finally, Section 6 concludes by pointing out possible generalizations and applications.

\section{General setting and preliminaries}	\label{sect-descr}

Let $(\Omega,\cA,P)$ be a given probability space, with $P$ denoting the physical/statistical probability measure, and let $T\in(0,\infty)$ be a fixed time horizon. We assume that all random variables and stochastic processes introduced in the following are measurable with respect to the $\sigma$-fields $\cA$ and $\cA\otimes\mathcal{B}([0,T])$, respectively. Let $\bF=(\cF_t)_{0\leq t \leq T}$ be a filtration on $(\Omega,\cA,P)$, assumed to satisfy the usual conditions of right-continuity and $P$-completeness. For a given stochastic process $Y=(Y_t)_{0\leq t \leq T}$ on $(\Omega,\cA,P)$ we denote by $\bF^Y=(\cF^Y_t)_{0\leq t \leq T}$ the right-continuous $P$-augmented natural filtration of $Y$.

The random change point is represented by a \emph{random time} $\tau$, i.e., an $\cA$-measurable random variable $\tau:\Omega\rightarrow\R_+$ which is not necessarily an $\bF$-stopping time. Furthermore, we let $W^1=(W^1_t)_{0\leq t\leq T}$ and $W^2 =(W^2_t)_{0\leq t\leq T}$ be two independent Brownian motions on $(\Omega,\cA,\bF,P)$ and, for $i=1,2$, we denote by $\bF^{W^i}$ the natural $P$-augmented filtration of $W^i$.

We consider a financial market with one risky asset and one riskless asset.
As usual in the literature, we take the riskless asset as the numéraire and directly pass to discounted quantities. We denote by $S=(S_t)_{0\leq t \leq T}$ the discounted price process of the risky asset and suppose that $S$ can be represented as follows, for some initial value $S_0\in(0,\infty)$:
\begin{align}	\label{S}
S = S_0\,\cE(X),
\end{align}
where $\cE(\cdot)$ denotes the stochastic exponential and $X$ is described by the SDE
\begin{align}
\nonumber \ud X_t & = \left( \ifnodefault{t} \mu^1(t, X_t) +  \ifdefault{t} \mu^2(t, X_t)\right)\dt \\
  & \quad + \ifnodefault{t} \sigma^1(t, X_t) \,\ud W^1_t
  + \ifdefault{t} \sigma^2(t, X_t) \left(\rho\,\ud W^1_t+\sqrt{1-\rho^2}\,\ud W^2_t\right)
  \label{model-2}\\
\nonumber X_0 & = 0.
\end{align}
where $\rho\in[-1,1]$ is a correlation parameter.
The well-posedness of the above SDE, as well as the existence and uniqueness of a solution, will be proved in Section \ref{bG} on a suitable filtered probability space (see Proposition \ref{solution}).
The functions $\mu^i:[0,T]\times\R\rightarrow\R$ and $\sigma^i:[0,T]\times\R\rightarrow(0,\infty)$, for $i=1,2$, are Borel-measurable and are assumed to satisfy the following condition.

\begin{condition}	\label{B}
The functions $\mu^i:[0,T]\times\R\rightarrow\R$ and $\sigma^i:[0,T]\times\R\rightarrow(0,\infty)$, for $i=1,2$, satisfy the following conditions:
\begin{itemize}
\item[(a)]
there exists a constant $K>0$ such that:
$$
|\mu^i(t,x)-\mu^i(t,y)| \leq K|x-y|,
\qquad \forall t\in[0,T], \,\forall x,y\in\R, \,\text{ for } i=1,2;
$$
$$
|\sigma^i(t,x)-\sigma^i(t,y)| \leq K|x-y|,
\qquad \forall t\in[0,T], \,\forall x,y\in\R, \,\text{ for } i=1,2;
$$
\item[(b)]
the function $(t,x)\mapsto\sigma^i(t,x)$ is jointly continuous in $(t,x)\in[0,T]\times\R$, for $i=1,2$.
\end{itemize}
\end{condition}

Part (a) of Condition \ref{B} consists of the usual global Lipschitz conditions on the functions $\mu^i$ and $\sigma^i$ appearing in the SDE \eqref{model-2}, while part (b) is needed for technical reasons.

\begin{remarks}	\label{growth}
\textbf{1)}
As can be easily verified, part (a) of Condition \ref{B} implies that there exists a constant $\bar{K}>0$ such that the usual growth conditions hold:
$$	\begin{gathered}
|\mu^i(t,x)|^2\leq\bar{K}\left(1+x^2\right),
\qquad \forall t\in[0,T], \forall x\in\R, \,\text{ for }i=1,2;	\\
\bigl(\sigma^i(t,x)\bigr)^2\leq\bar{K}\left(1+x^2\right),
\qquad \forall t\in[0,T], \forall x\in\R, \,\text{ for }i=1,2.
\end{gathered} $$

\textbf{2)}
Observe that $\bF^S$ coincides with the filtration $\bF^X$ generated by the process $X$, meaning that $\cF^S_t=\cF^X_t$ for all $t\in[0,T]$. Indeed, due to \eqref{S}, it is evident that $\cF^S_t\subseteq\cF^X_t$. On the other hand, we have $X_t=x+\int_0^t\!S_u^{-1}\ud S_u$ and, hence, we also have $\cF^X_t\subseteq\cF^S_t$ for all $t\in[0,T]$.
\end{remarks}

As mentioned in the introduction, we aim at studying the properties of the model \eqref{S}-\eqref{model-2} with respect to different levels of information, mathematically represented by different filtrations on $(\Omega,\cA,P)$. In view of part 2 of Remarks \ref{growth}, we can and do restrict our attention to the study of the behavior of the process $X$ in the following filtrations:
\begin{itemize}
\item[\textbf{(i)}]
the filtration $\bF^X$;
\item[\textbf{(ii)}]
the filtration $\bG^X$, obtained as the \emph{progressive enlargement} of $\bF^X$ with respect to $\tau$ and defined as $\cG^X_t:=\bigcap_{s>t}\bigl(\cF^X_s\vee\sigma(\tau\wedge s)\bigr)$ for all $t\in[0,T]$;
\item[\textbf{(iii)}]
the filtration $\bG$, obtained as the \emph{progressive enlargement} of $\bF$ with respect to $\tau$ and defined as $\cG_t:=\bigcap_{s>t}\bigl(\cF_s\vee\sigma(\tau\wedge s)\bigr)$ for all $t\in[0,T]$;
\item[\textbf{(iv)}]
the filtration $\bG^{X, (\tau)}$, obtained as the \emph{initial enlargement} of $\bF^X$ with respect to $\tau$ and defined as $\cG^{X, (\tau)}_t:=\bigcap_{s>t}\bigl(\cF^X_s\vee\sigma(\tau)\bigr)$ for all $t\in[0,T]$;
\item[\textbf{(v)}]
the filtration $\bG^{(\tau)}$, obtained as the \emph{initial enlargement} of $\bF$ with respect to $\tau$ and defined as $\cG_t^{(\tau)}:=\bigcap_{s>t}\bigl(\cF_s\vee\sigma(\tau)\bigr)$ for all $t\in[0,T]$.
\end{itemize}

The filtrations $\bG^X$ and $\bG$ are the smallest right-continuous filtrations which contain $\bF^X$ and $\bF$, respectively, and make $\tau$ a $\bG^X$-stopping time and a $\bG$-stopping time, respectively. For a detailed account of the theory of enlargement of filtrations we refer the reader to the monograph of \citeN{Jeulin80} and to Chapter VI of \citeN{Protter05}.
It is easy to see that:
$$
\bF^X\subseteq\bG^X\subseteq\bG\subseteq\bG^{(\tau)},
$$
meaning that $\cF^X_t\subseteq\cG^X_t\subseteq\cG_t\subseteq\cG^{(\tau)}_t$ for all $t\in[0,T]$. It is also evident that
$$
\bG^X\subseteq\bG^{X, (\tau)}\subseteq\bG^{(\tau)}.
$$

Intuitively, in the special case where $\bF=\bF^{W^1}\vee\bF^{W^2}$, the different filtrations introduced above correspond to market participants having access to different information sets:
\begin{itemize}
\item[\textbf{(i)}]
$\cF^X_t$: the knowledge of only the price process of the risky asset up to time $t$;
\item[\textbf{(ii)}]
$\cG^X_t$: the knowledge of the price process of the risky asset up to time $t$ plus the knowledge of the random time $\tau$ if the latter has occurred before time $t$;
\item[\textbf{(iii)}]
$\cG_t$: the knowledge of the two driving Brownian motions $W^1$ and $W^2$ up to time $t$ plus the knowledge of the random time $\tau$ if the latter has occurred before time $t$;
\item[\textbf{(iv)}]
$\cG^{X,(\tau)}_t$: the knowledge of the price process of the risky asset up to time $t$ plus the knowledge (already at time $t=0$) of the random time $\tau$;
\item[\textbf{(v)}]
$\cG^{(\tau)}_t$: the knowledge of the two driving Brownian motions $W^1$ and $W^2$ up to time $t$ plus the knowledge (already at time $t=0$) of the random time $\tau$.
\end{itemize}

We shall denote by $^pY$ the predictable projection of a process $Y=(Y_t)_{0\leq t \leq T}$ onto one of the filtrations introduced above (see e.g. \citeN{HeWangYan92}, Section V.1). The filtration onto which we will take projections changes throughout the paper, but this will be made clear in the text.

Let $\bA=(\cA_t)_{0\leq t \leq T}$ be a generic filtration on $(\Omega,\cA,P)$ with respect to which the process $S$ is a semimartingale and $L(S,\bA)$ be the set of all $S$-integrable $\bA$-predictable processes, in the sense of Definition 9.13 in \citeN{HeWangYan92}.
We denote by $\int\!h\,\ud S$ the stochastic integral process $\bigl(\int_0^t\!h_u\,\ud S_u\bigr)_{0\leq t \leq T}$, for $h=(h_t)_{0\leq t \leq T}\in L(S,\bA)$, and by $\Mloc(\bA)$ the family of all $\bA$-local martingales.
Note that, if $S=M+B$ denotes the canonical decomposition of $S$ into $M\in\Mloc(\bA)$ and a continuous $\bA$-predictable process of finite variation $B$, we have $\int\!h\,\ud S=\int\!h\,\ud M+\int\!h\,\ud B$ (see \citeN{HeWangYan92}, Theorem 9.16).

In order to study the no-arbitrage properties of the model \eqref{S}-\eqref{model-2}, we recall the characterizations of two important notions of arbitrage which have been considered in the literature.

\pagebreak

\begin{definition}	\label{def-arb}
\mbox{}
\begin{itemize}
\item[(i)]
We say that \emph{No Arbitrage of the First Kind (NA1)} holds in the filtration $\bA$ if there exists an \emph{$\bA$-local martingale deflator}, i.e., a process $L=(L_t)_{0\leq t \leq T}\in\Mloc(\bA)$ with $L_0=1$ and $L_T>0$ $P$-a.s. and such that $SL\in\Mloc(\bA)$;
\item[(ii)]
We say that \emph{No Free Lunch with Vanishing Risk (NFLVR)} holds in the filtration $\bA$ if there exists an \emph{$\bA$-martingale deflator}, i.e., a process $L=(L_t)_{0\leq t \leq T}\in\Mloc(\bA)$ with $L_0=1$, $L_T>0$ $P$-a.s. and $E[L_T]=1$ and such that $SL\in\Mloc(\bA)$.
\end{itemize}
\end{definition}

Part (i) of Definition \ref{def-arb} is due to \citeN{Kardaras10}, while part (ii) goes back to \citeN{DelbaenSchachermayer94}. In particular, the NA1 condition is weaker than NFLVR and, moreover, can be shown to be the minimal condition for market viability. Note also that martingale deflators correspond to density processes of Equivalent Local Martingale Measures (ELMMs).
We refer to \citeN{FontanaRunggaldier12} for a study of the two no-arbitrage conditions introduced above in the context of general diffusion-based models.

\section{The progressively enlarged filtrations $\bG$ and $\bG^X$}

In this section we study the progressively enlarged filtrations $\bG$ and $\bG^X$. We shall make no assumption on the random time $\tau$ apart from a very weak semimartingale-preservation hypothesis (Condition \ref{C}). We start our analysis with the progressively enlarged filtration $\bG$, which is easier to describe than the filtration $\bG^X$. Moreover, starting with the filtration $\bG$ allows us to prove the well-posedness of the SDE \eqref{model-2}.

\subsection{The progressively enlarged filtration $\bG$}	\label{bG}

The filtration $\bG$ is the smallest filtration satisfying the usual conditions which contains $\bF$ and makes $\tau$ a $\bG$-stopping time. 
However, the $\bF$-Brownian motions $W^1$ and $W^2$ may fail to be $\bG$-semimartingales. The following condition prevents this pathological behavior\footnote{We want to point out that, due to Proposition 4.16 of \citeN{Jeulin80} together with the Kunita-Watanabe inequality, Condition \ref{C} is always satisfied for $t\leq\tau$.}.

\begin{condition}	\label{C}
There exist two $\bG$-predictable processes $\theta^1=(\theta^1_t)_{0\leq t \leq T}$ and $\theta^2=(\theta^2_t)_{0\leq t \leq T}$ and two $\bG$-Brownian motions $\Wtilde^1=(\Wtilde^1_t)_{0\leq t  \leq T}$ and $\Wtilde^2=(\Wtilde^2_t)_{0\leq t \leq T}$ such that:
$$
W^i_t = \Wtilde^i_t+\int_0^t\!\theta^i_u\,\ud u,
\qquad\quad\text{for all }t\in[0,T]\text{ and for }i=1,2.
$$
\end{condition}

Condition \ref{C} can be regarded as a rather weak form of the \emph{($\mathcal{H}'$)-hypothesis} from the theory of enlargement of filtrations, which assumes that all $\bF$-semimartingales are also $\bG$-semimartingales (\citeN{Jeulin80}, Chapter II). Condition \ref{C} can be shown to hold for almost all random time models considered in financial and insurance mathematics (in particular, it is always trivially satisfied in the common case when $\tau$ is a doubly stochastic random time, see Section \ref{immersion}, as well as when the density hypothesis holds, see e.g. \citeN{ElKarouiJeanblancJiao10}).

\subsubsection{Existence and uniqueness of the solution to the SDE \eqref{model-2}}

As long as Condition \ref{C} holds, equation \eqref{model-2} makes sense as a semimartingale-driven SDE on  $(\Omega,\cA,\bG,P)$. This provides a good setting for establishing the existence and uniqueness of a solution. We say that a $\bG$-semimartingale $X=(X_t)_{0\leq t \leq T}$ is a \emph{solution} to the SDE \eqref{model-2} on $(\Omega,\cA,\bG,P)$ if $X_0=0$ and $X$ satisfies equation \eqref{model-2} with respect to the $\bG$-semimartingales $W^1$ and $W^2$. This corresponds to the notion of \emph{strong solution} of a semimartingale-driven SDE, as considered in Chapter V of \citeN{Protter05} (see also \citeN{Jacod79}, Chapter XIV).

\begin{proposition}	\label{solution}
Suppose that Conditions \ref{B} and \ref{C} hold. Then there exists a unique continuous $\bG$-semimartingale $X=(X_t)_{0\leq t \leq T}$ which is a solution to the SDE \eqref{model-2} on $(\Omega,\cA,\bG,P)$.
\end{proposition}
\begin{proof}
Since $\tau$ is a $\bG$-stopping time, the processes $\indik_{\lsi0,\tau\rsi}$ and $\indik_{\lsir\tau,T\rsi}$ are $\bG$-predictable, being $\bG$-adapted and left-continuous, and admit limits from the right. Let us define the following random functions, for $\omega\in\Omega$, $t\in[0,T]$, $x\in\R$:
\be	\label{rnd-fcts}	\begin{aligned}
g(\omega,t,x) &:= \indik_{\{t\leq\tau(\omega)\}}\mu^1(t,x)+\indik_{\{t>\tau(\omega)\}}\mu^2(t,x);	\\
f(\omega,t,x) &:= \indik_{\{t\leq\tau(\omega)\}}\sigma^1(t,x)+\indik_{\{t>\tau(\omega)\}}\sigma^2(t,x).
\end{aligned} \ee
Condition \ref{B} implies that $f$ and $g$ are random Lipschitz, in the sense of \citeN{Protter05}, page 256. The existence of a unique solution $X=(X_t)_{0\leq t \leq T}$ to the SDE \eqref{model-2} on $(\Omega,\cA,\bG,P)$ then follows from Theorem V.6 of \citeN{Protter05}.
\end{proof}

\subsubsection{Canonical decomposition and no-arbitrage properties in $\bG$}

Let us now investigate the no-arbitrage properties of the financial market where the asset $S$ is traded with respect to the information contained in the progressively enlarged filtration $\bG$. As a preliminary, we write the canonical decomposition of the process $X=(X_t)_{0\leq t \leq T}$ in the filtration $\bG$:

\be	\label{candec-G}
X_t = \int_0^t\!\mutilde_u\,\ud u+\int_0^t\!V_u\,\ud\Wtilde_u,
\qquad \text{for all }t\in[0,T],
\ee
where the processes $\mutilde=(\mutilde_t)_{0\leq t \leq T}$, $V=(V_t)_{0\leq t \leq T}$ and $\Wtilde=(\Wtilde_t)_{0\leq t \leq T}$ are defined as\footnote{Note that Condition \ref{C} implicitly requires that $\int_0^T\!|\theta^i_u|\ud u<\infty$ $P$-a.s., for $i=1,2$. In turn, due to Condition \ref{B}-(b) together with the continuity of $X$, this implies that $\int_0^t\!\theta^i_u\,\sigma^i(u,X_u)\,\ud u$ is well-defined, for all $t\in[0,T]$ and $i=1,2$.}:
\be	\label{candec-G-1}
\mutilde_t :=
\indik_{\{t\leq\tau\}}\left(\mu^1(t,X_t)+\sigma^1(t,X_t)\theta^1_t\right)
+\indik_{\{t>\tau\}}\left(\mu^2(t,X_t)+\sigma^2(t,X_t)\bigl(\rho\theta^1_t+\sqrt{1-\rho^2}\theta^2_t\bigr)\right);
\ee
\be	\label{candec-G-2}
V_t := \indik_{\{t\leq\tau\}}\sigma^1(t,X_t)+\indik_{\{t>\tau\}}\sigma^2(t,X_t);
\ee
\begin{gather}
\Wtilde_t := 
\Wtilde^1_{t\wedge\tau}+\rho\,(\Wtilde^1_{t\vee\tau}-\Wtilde^1_{\tau})+\sqrt{1-\rho^2}\,(\Wtilde^2_{t\vee\tau}-\Wtilde^2_{\tau})
\label{candec-G-3}
\end{gather}

Since $[\Wtilde]_t=t$, for all $t\in[0,T]$, the continuous $\bG$-local martingale $\Wtilde$ is a $\bG$-Brownian motion. Equation \eqref{candec-G} gives the canonical decomposition of $X$ in $\bG$ and leads to the next proposition, which characterizes the no-arbitrage properties of the model \eqref{S}-\eqref{model-2} in the progressively enlarged filtration $\bG$. We define the $\bG$-predictable process $\bar{\theta}=(\bar{\theta}_t)_{0\leq t \leq T}$ as:
\be	\label{bartheta}
\bar{\theta} := \indik_{\lsi0,\tau\rsi}\theta^1
+\indik_{\lsir\tau,T\rsi}\bigl(\rho\theta^1+\sqrt{1-\rho^2}\theta^2\bigr).
\ee

\begin{proposition}	\label{arb-G}
Suppose that Conditions \ref{B} and \ref{C} hold. Then the following assertions hold for the model \eqref{S}-\eqref{model-2} considered with respect to the filtration $\bG$:
\begin{itemize}
\item[(a)]
NA1 holds if and only if $\int_0^T\!\bar{\theta}_t^2\,\ud t<\infty$ $P$-a.s., with the latter condition being equivalent to $\int_0^T\!(\mutilde_t/V_t)^2\,\ud t<\infty$ $P$-a.s.;
\item[(b)]
NFLVR holds if and only if NA1 holds and there exists $N=(N_t)_{0\leq t \leq T}\in\tMloc(\bG)$ with $N_0=0$, $\Delta N>-1$ $P$-a.s., $[N,\Wtilde]=0$ such that  $E\bigl[\mathcal{E}(-\int(\mutilde/V)\,\ud\Wtilde+N)_T\bigr]=1$.
\end{itemize}
\end{proposition}
\begin{proof}
Note first that, due to Conditions \ref{B}-\ref{C}, the process $\mutilde/V$ is well-defined. In view of part (i) of Definition \ref{def-arb} together with Theorem 4 of \citeN{Kardaras10}, NA1 holds if and only if the following condition holds:
$$
\int_0^T\!\left(\frac{\mutilde_t}{V_t}\right)^2\!\ud t
= \int_0^{T\wedge\tau}\!\left(\frac{\mu^1(t,X_t)}{\sigma^1(t,X_t)}+\theta^1_t\right)^2\!\ud t\,
+\int_{\tau}^{T\vee\tau}\!\left(\frac{\mu^2(t,X_t)}{\sigma^2(t,X_t)}+\rho\,\theta^1_t+\sqrt{1-\rho^2}\,\theta^2_t\right)^2\!\ud t
<\infty \text{ $P$-a.s.}
$$
Due to Condition \ref{B}, the continuity of the function $\sigma^i:[0,T]\times\R\rightarrow (0, \infty)$, for $i=1,2$, together with the continuity of $X$, implies that $\xi:=\min_{t\in[0,T]}\bigl\{\sigma^1(t,X_t)\wedge\sigma^2(t,X_t)\bigr\}$ is well-defined and $P$-a.s. strictly positive.
Hence, by using the elementary inequality $(a+b)^2\leq 2a^2+2b^2$ together with Remarks \ref{growth}-1, we can write:
$$	\begin{aligned}
\int_0^T\!\left(\frac{\mutilde_t}{V_t}\right)^2\!\ud t
&\leq 2\left(\int_0^{T\wedge\tau}\!\left(\frac{\mu^1(t,X_t)}{\sigma^1(t,X_t)}\right)^2\!\ud t
+\int_{\tau}^{T\vee\tau}\!\left(\frac{\mu^2(t,X_t)}{\sigma^2(t,X_t)}\right)^2\!\ud t\right)
+2\int_0^T\!\bar{\theta}^2_t\,\ud t	\\
&\leq \frac{2\,\bar{K}}{\xi^2}\int_0^T(1+X^2_t)\,\ud t
+2\int_0^T\!\bar{\theta}^2_t\,\ud t	
\leq \frac{2\,\bar{K}\,T}{\xi^2}\Bigl(1+\underset{t\in[0,T]}{\max}X^2_t\Bigr)+2\int_0^T\!\bar{\theta}^2_t\,\ud t.
\end{aligned} $$
Analogously, using the elementary inequality $(a+b)^2\geq b^2/2-a^2$:
$$	\begin{aligned}
\int_0^T\!\left(\frac{\mutilde_t}{V_t}\right)^2\!\ud t
&\geq \frac{1}{2}\int_0^T\!\bar{\theta}^2_t\,\ud t
-\int_0^{T\wedge\tau}\!\left(\frac{\mu^1(t,X_t)}{\sigma^1(t,X_t)}\right)^2\!\ud t
-\int_{\tau}^{T\vee\tau}\!\left(\frac{\mu^2(t,X_t)}{\sigma^2(t,X_t)}\right)^2\!\ud t	\\	
&\geq \frac{1}{2}\int_0^T\!\bar{\theta}^2_t\,\ud t
-\frac{\bar{K}\,T}{\xi^2}\Bigl(1+\underset{t\in[0,T]}{\max}X^2_t\Bigr).
\end{aligned}	$$
Since $X$ is continuous, the above inequalities show that $\int_0^T\!(\mutilde_t/V_t)^2\,\ud t<\infty$ $P$-a.s. if and only if $
\int_0^T\!\!\bar{\theta}^2_t\,\ud t<\infty$ $P$-a.s., thus proving part (a).
Part (b) can then be easily proved by relying on part (ii) of Definition \ref{def-arb} together with Lemma 4.3.15 of \citeN{FontanaRunggaldier12}.
\end{proof}

\begin{remark}
Note that, if the Brownian motions $W^1$ and $W^2$ are $\bG$-semimartingales but their finite variation parts are not absolutely continuous with respect to $\ud t$ (i.e., Condition \ref{C} is violated), one can then obtain the most egregious form of arbitrage (i.e., an \emph{increasing profit}) in the filtration $\bG$, see e.g. Section 4.3 of \citeN{FontanaRunggaldier12}. In this sense, Condition \ref{C} is minimal for the study of the no-arbitrage properties of the model \eqref{S}-\eqref{model-2}.
\end{remark}

Proposition \ref{arb-G} shows that the process  $\bar{\theta}$ defined in \eqref{bartheta} plays a crucial role in determining the no-arbitrage properties of the model \eqref{S}-\eqref{model-2} in $\bG$. In turn, this implies that the existence of arbitrages in $\bG$ crucially depends on the properties of $\tau$. For instance, the condition $\int_0^T\!\bar{\theta}^2_t\,\ud t<\infty$ $P$-a.s. may fail in the cases considered in \citeN{Imkeller02}.

\begin{remark}[\emph{The $\bG$-martingale representation property}]
Let us denote by $A^{\bG}$ the $\bG$-predictable compensator of the increasing process $(\indik_{\{\tau\leq t\}})_{0\leq t \leq T}$ (see \citeN{Protter05}, Section III.5) and let $M^{\bG}:=\indik_{\{\tau\leq\cdot\}}-A^{\bG}$ be the corresponding compensated $\bG$-martingale.
Suppose that every $\bG$-local martingale $L=(L_t)_{0\leq t \leq T}$ admits the representation:
\be	\label{mrp-1}
L_t = L_0 + \int_0^t\!\varphi^1_u\,\ud\Wtilde^1_u + \int_0^t\!\varphi^2_u\,\ud\Wtilde^2_u+ \int_0^t\!\psi_u\,\ud M^{\bG}_u,
\qquad \text{for all }t\in[0,T],
\ee
where $\varphi^i=(\varphi^i_t)_{0\leq t \leq T}$ is a $\bG$-predictable process with $\int_0^T(\varphi^i_t)^2\ud t<\infty$ $P$-a.s., for $i=1,2$, and $\psi=(\psi_t)_{0\leq t \leq T}$ is a $\bG$-predictable process with $\int_0^T\!|\psi_t||\ud A^{\bG}_t|<\infty$ $P$-a.s. and where $\Wtilde^i$, for $i=1,2$, is the $\bG$-Brownian motion introduced in Condition \ref{C}. In this case, we can obtain a more precise description of the $\bG$-local martingale $N=(N_t)_{0\leq t \leq T}$ appearing in part (b) of Proposition \ref{arb-G}. Indeed, using \eqref{candec-G}-\eqref{candec-G-3} together with \eqref{mrp-1}, we get:
\be	\label{N}
N_t = \int_0^{t\wedge\tau}\!\!\varphi^2_u\,\ud\Wtilde^2_u
+\indik_{\{\rho=0\}}\!\!\int_{\tau}^{t\vee\tau}\!\!\varphi^1_u\,\ud\Wtilde^1_u
+\indik_{\{\rho\neq0\}}\!\Biggl(\int_{\tau}^{t\vee\tau}\!\!\varphi^3_u\,\ud\Wtilde^2_u
-\frac{\sqrt{1-\rho^2}}{\rho}\int_{\tau}^{t\vee\tau}\!\!\varphi^3_u\,\ud\Wtilde^1_u\Biggr)
+\int_0^t\!\!\psi_u\,\ud M^{\bG}_u,
\ee
where $\varphi^1$, $\varphi^2$, $\varphi^3$ and $\psi$ are $\bG$-predictable integrable processes. As long as NFLVR and the representation property \eqref{mrp-1} hold in $\bG$, part (b) of Proposition \ref{arb-G} and \eqref{N} give a complete characterization of all $\bG$-martingale deflators (or, equivalently, of all ELMMs in $\bG$).
If $\bF=\bF^{W^1}\vee\bF^{W^2}$, necessary and sufficient conditions for the validity of the martingale representation property \eqref{mrp-1} have been recently established by \citeN{JeanblancSong12}. 
\end{remark}

\subsection{The progressively enlarged price process filtration $\bG^X$}	\label{bG^S}

This section studies the properties of the model \eqref{S}-\eqref{model-2} when considered with respect to the filtration $\bG^X$.

\subsubsection{Canonical decomposition and no-arbitrage properties in $\bG^X$}

The next lemma gives the canonical decomposition of $X$ with respect to the filtration $\bG^X$. The idea of the proof consists of projecting the canonical decomposition \eqref{candec-G} obtained with respect to $\bG$ onto the smaller filtration $\bG^X$. In order to take care of integrability issues, a localization procedure is needed.

\begin{lemma}	\label{candec-G^X}
Suppose that Conditions \ref{B} and \ref{C} hold. Then the process $X$ admits the following canonical decomposition with respect to the filtration $\bG^X$:
\be	\label{candec-G^X-1}
X_t = \int_0^t\!\mu_u\,\ud u+\int_0^t\!V_u\,\ud B_u,
\qquad \text{for all }t\in[0,T],
\ee
where the $\bG^X$-predictable process $\mu=(\mu_t)_{0\leq t \leq T}$ is defined as:
\be	\label{candec-G^X-2}
\mu_t := \indik_{\{t\leq\tau\}}\!\left(\mu^1(t,X_t)+\sigma^1(t,X_t)\,^p\theta^1_t\right)
+\indik_{\{t>\tau\}}\!\left(\mu^2(t,X_t)
+\sigma^2(t,X_t)\bigl(\rho\,^p\theta^1_t+\sqrt{1-\rho^2}\,^p\theta^2_t\bigr)\right),
\ee
with $^p\theta^i$ denoting the $\bG^X$-predictable projection of $\theta^i$, for $i=1,2$, and where the process $B=(B_t)_{0\leq t \leq T}$ is a $\bG^X$-Brownian motion and the $\bG^X$-predictable process $V=(V_t)_{0\leq t \leq T}$ is defined as in \eqref{candec-G-2}.
\end{lemma}
\begin{proof}
The unique solution $X=(X_t)_{0\leq t \leq T}$ to the SDE \eqref{model-2} on $(\Omega,\cA,\bG,P)$ is a continuous $\bG$-semimartingale. Hence, the result of \citeN{Stricker77} implies that it is also a $\bG^X$-semimartingale, since $\bG^X\subset\bG$. Let $X=A+M$ be the canonical decomposition of $X$ in $\bG^X$, where $A$ is a continuous $\bG^X$-predictable process of finite variation with $A_0=0$ and $M\in\Mloc(\bG^X)$ with $M_0=0$. For every $n\in\N$, define the $\bG^X$-stopping time $\tau_n:=\inf\{t\in[0,T]:|X_t|\geq n\}\wedge T$. Clearly, we have $\tau_n\nearrow T$ $P$-a.s. as $n\rightarrow\infty$. Due to Remark \ref{growth}-1, there exists a constant $\bar{K}>0$ such that, for every $n\in\N$:
$$
\int_0^{\tau_n\wedge T}\!\!V^2_t\,\ud t
\leq \bar{K}\!\!\int_0^{\tau_n\wedge T}\!\!\left(1+X^2_t\right)\ud t
\leq \bar{K}(1+n^2)\,T
\qquad\text{$P$-a.s.}
$$
Since $[M]_t=\int_0^tV^2_u\,\ud u$, for all $t\in[0,T]$, this shows that $\{\tau_n\}_{n\in\N}$ is a $\bG^X$-localizing sequence for the $\bG^X$-local martingale $M$ as well as for the $\bG$-local martingale $\int\!V\ud\Wtilde$. Proposition 9.24 of \citeN{Jacod79} together with \eqref{candec-G} then implies that $A$ is given by the dual $\bG^X$-predictable projection of the process $\int_0^{\cdot}\mutilde_t\,\ud t$, i.e. $A=\left(\int_0^{\cdot}\mutilde_t\,\ud t\right)^p$. Furthermore, as is shown in 1.40 of \citeN{Jacod79}, we have $\left(\int_0^{\cdot}\mutilde_t\,\ud t\right)^p=\int_0^{\cdot}\!\,^p\mutilde_t\,\ud t$, where $^p\mutilde$ denotes the $\bG^X$-predictable projection of $\mutilde$. Equation \eqref{candec-G^X-2} then follows by noting that, since $\tau$ is a $\bG^X$-stopping time, the processes $\indik_{\lsi0,\tau\rsi}$ and $\indik_{\lsir\tau,T\rsi}$ are $\bG^X$-adapted and left-continuous and, hence, $\bG^X$-predictable.
To finish the proof, the process $V$ defined in \eqref{candec-G-2} never hits zero and is $\bG^X$-predictable, $\tau$ being a $\bG^X$-stopping time. This implies that the stochastic integral $B:=\int\!V^{-1}\,\ud M$ is well-defined as a continuous $\bG^X$-local martingale with $B_0=0$. Since $[B]_t=\int_0^t\!V^{-2}_u\,\ud[M]_u=t$, for all $t\in[0,T]$, Lévy's characterization theorem allows to conclude that $B$ is a $\bG^X$-Brownian motion.
\end{proof}

We can now answer the question of whether the model \eqref{S}-\eqref{model-2}, considered with respect to the filtration $\bG^X$, allows for arbitrage profits. We denote by $^p\bar{\theta}$ the $\bG^X$-predictable projection of the process $\bar{\theta}=\indik_{\lsi0,\tau\rsi}\theta^1+\indik_{\lsir\tau,T\rsi}\bigl(\rho\,\theta^1+\sqrt{1-\rho^2}\,\theta^2\bigr)$ introduced in \eqref{bartheta}.
The proof of the next proposition is similar to the proof of Proposition \ref{arb-G} and, hence, omitted.

\begin{proposition}	\label{arb-G^X}
Suppose that Conditions \ref{B} and \ref{C} hold. Then the following assertions hold for the model \eqref{S}-\eqref{model-2} considered with respect to the filtration $\bG^X$:
\begin{itemize}
\item[(a)]
NA1 holds if and only if $\int_0^T(^p\bar{\theta}_t)^2\,\ud t<\infty$ $P$-a.s., with the latter condition being equivalent to $\int_0^T\!(\mu_t/V_t)^2\,\ud t<\infty$ $P$-a.s.
\item[(b)]
NFLVR holds if and only if NA1 holds and there exists $N=(N_t)_{0\leq t \leq T}\in\tMloc(\bG^X)$ with $N_0=0$, $\Delta N>-1$ $P$-a.s., $[N,B]=0$ such that $E\left[\cE\left(-\int(\mu/V)\ud B+N\right)_T\right]=1$.
\end{itemize}
\end{proposition}

\subsubsection{Martingale representation property in $\bG^X$}	\label{MRP-G^X}

We now study in more detail the structure of the filtration $\bG^X$. In particular, we aim at proving a martingale representation result (see Proposition \ref{MRP}). In turn, this will lead to an explicit characterization of all $\bG^X$-martingale deflators (see Corollary \ref{NFLVR-G^X}).

As a preliminary, observe that the process $1/V$ is well-defined, $\bG^X$-predictable and locally bounded, being left-continuous by part (b) of Condition \ref{B}. Hence, we can define the $\bG^X$-adapted continuous process $\What=(\What_t)_{0\leq t \leq T}$ as follows, for all $t\in[0,T]$:
\be	\label{def-Yhat}
\What_t := \int_0^tV_u^{-1}\,\ud X_u=\int_0^t\frac{\mu_u}{V_u}\,\ud u+B_t,
\ee
where the processes $\mu$ and $B$ are as in Lemma \ref{candec-G^X}. Let us denote by $\bF^{\What}=(\cF^{\What}_t)_{0\leq t \leq T}$ the right-continuous $P$-augmented natural filtration of $\What$ and by $\bG^{\What}=(\cG^{\What}_t)_{0\leq t \leq T}$ the progressive enlargement of $\bF^{\What}$ with respect to $\tau$, meaning that $\cG^{\What}_t=\bigcap_{s>t}\bigl(\cF^{\What}_s\vee\sigma(\tau\wedge s)\bigr)$ for all $t\in[0,T]$. We can now prove a useful lemma which describes the structure of the filtration $\bG^X$, showing that it coincides with the progressive enlargement with respect to $\tau$ of the filtration generated by the drifted Brownian motion $\What$.

\begin{lemma}	\label{filtrations}
Suppose that Conditions \ref{B} and \ref{C} hold. Then $\bG^X=\bG^{\What}$.
\end{lemma}
\begin{proof}
Clearly, the process $\What$ defined in \eqref{def-Yhat} is $\bG^X$-adapted and $\tau$ is a $\bG^X$-stopping time. This implies that $\bG^{\What}\subseteq\bG^X$. To prove the converse inclusion, let us first note that the process $X$ can be represented as follows, where the random function $f$ is defined as in \eqref{rnd-fcts}:
$$
X_t = \int_0^t\!f(u,X_u)\,\ud\What_u,
\qquad\text{for all }t\in[0,T].
$$
Let us define the process $X^0=(X^0_t)_{0\leq t\leq T}$ by $X^0_t:=0$ for all $t\in[0,T]$ and define inductively, for every $k\in\N$, the process $X^k=(X^k_t)_{0\leq t \leq T}$ as:
$$
X^{k+1}_t := \int_0^t\!f(u,X^k_u)\,\ud\What_u,
\qquad \text{for all }t\in[0,T].
$$
By construction, for every $k\in\N$, the process $X^k$ is adapted to the filtration $\bG^{\What}$. Furthermore, since the function $f$ is random Lipschitz, Theorem V.8 of \citeN{Protter05} implies that the process $X^k$ converges to $X$ uniformly on compacts in probability and, up to a subsequence, the convergence takes place $P$-a.s. uniformly on compacts. This implies that $X$ is adapted to the filtration $\bG^{\What}$, thus showing that $\bG^X\subseteq\bG^{\What}$.
\end{proof}

We are now in position to prove a martingale representation result for the filtration $\bG^X$. It is important to note that we do not make any assumption on $\tau$ nor on the underlying filtration $\bF$ (in particular, we do not assume that $\bF=\bF^{W^1}\vee\bF^{W^2}$). We denote by $A^{\bG^X}$ the $\bG^X$-predictable compensator of $\tau$ and by $M^{\bG^X}:=\indik_{\{\tau\leq\cdot\}}-A^{\bG^X}$ the associated $\bG^X$-martingale.

\begin{proposition}	\label{MRP}
Suppose that Conditions \ref{B} and \ref{C} hold and assume in addition that NA1 holds in the filtration $\bG^X$. Then every $\bG^X$-local martingale $L=(L_t)_{0\leq t \leq T}$ admits a representation of the form:
\be	\label{MRP-0}
L_t = L_0+\int_0^t\!\varphi_u\,\ud B_u+\int_0^t\!\psi_u\,\ud M^{\bG^X}_u,
\qquad\text{for all }t\in[0,T],
\ee
for some $\bG^X$-predictable processes $\varphi=(\varphi_t)_{0\leq t \leq T}$ and $\psi=(\psi_t)_{0\leq t \leq T}$ with $\int_0^T\!\varphi_t^2\,\ud t<\infty$ $P$-a.s. and $\int_0^T\!|\psi_t||\ud A^{\bG^X}_t|<\infty$ $P$-a.s.
\end{proposition}
\begin{proof}
As in Proposition \ref{arb-G^X}, NA1 holds if and only if $\int_0^T\!(\mu_t/V_t)^2\,\ud t<\infty$ $P$-a.s. Hence, we can define the strictly positive continuous $\bG^X$-local martingale $\Zhat:=\cE(-\int(\mu/V)\,\ud B)$. Let $\{\tau_n\}_{n\in\N}$ be a localizing sequence for $\Zhat$, meaning that $\Zhat^{\tau_n}$ is a uniformly integrable $\bG^X$-martingale, for all $n\in\N$. For every $n\in\N$, define the filtration $\bG^{X,n}:=(\cG^X_{t\wedge\tau_n})_{0\leq t \leq T}$ and, analogously, $\bF^{\What,n}:=(\cF^{\What}_{t\wedge\tau_n})_{0\leq t \leq T}$ and $\bG^{\What,n}:=(\cG^{\What}_{t\wedge\tau_n})_{0\leq t \leq T}$. For every $n\in\N$, let the probability measure $Q^n$ be defined on $\cG^X_{T\wedge\tau_n}$ by $dQ^n=\Zhat_{T\wedge\tau_n}dP$. Girsanov's theorem implies that $(M^{\bG^X})^{\tau_n}$ is a $(Q^n,\bG^{X,n})$-martingale, since $[\Zhat,M^{\bG^X}]^{\tau_n}=0$, and also that the stopped process $\What^{\tau_n}$ is also a continuous $(Q^n,\bG^{X,n})$-local martingale, for all $n\in\N$. Indeed, for all $t\in[0,T]$:
$$
\What^{\tau_n}_t
= B^{\tau_n}_t+\int_0^{t\wedge\tau_n}\!\frac{\mu_u}{V_u}\,\ud u
= B^{\tau_n}_t-\int_0^{t\wedge\tau_n}\!\frac{1}{\Zhat^{\tau_n}_u}\,\ud [\Zhat,B]_u^{\tau_n}.
$$
In particular, since $[\What^{\tau_n}]_t=[B]_t^{\tau_n}=\tau_n\wedge t$, for all $t\in[0,T]$, the process $\What^{\tau_n}$ is a stopped $(Q^n,\bF^{\What,n})$-Brownian motion. As a consequence, every $(Q^n,\bF^{\What,n})$-local martingale can be represented as a stochastic integral of $\What^{\tau_n}$.
Since $\What^{\tau_n}$ is also a stopped $(Q^n,\bG^{X,n})$-Brownian motion, all $(Q^n,\bF^{\What,n})$-local martingales are also $(Q^n,\bG^{X,n})$-local martingales. Hence, Theorem 2.3 of \citeN{Kusuoka99} together with Lemma \ref{filtrations} implies that any $(Q^n,\bG^{X,n})$-local martingale $\hat{L}=(\hat{L}_t)_{0\leq t \leq T}$ can be represented as:
\be	\label{MRP-1}
\hat{L}_t = \hat{L}_0+\int_0^t\!\varphi_u\,\ud\What^{\tau_n}_u+\int_0^t\!\psi_u\,\ud (M^{\bG^X})^{\tau_n}_u,
\qquad\text{for all }t\in[0,T],
\ee
where $\varphi=(\varphi_t)_{0\leq t \leq T}$ and $\psi=(\psi_t)_{0\leq t \leq T}$ are two $\bG^{X,n}$-predictable processes such that $\int_0^T\!\varphi^2_t\,\ud t<\infty$ $P$-a.s. and $\int_0^T\!|\psi_t||\ud A^{\bG^X}_t\!|<\infty$ $P$-a.s.
Let $L\in\Mloc(P,\bG^X)$. By Girsanov's theorem, the difference $L^{\tau_n}-\int\!\frac{1}{\Zhat^{\tau_n}}\,\ud[\Zhat,L]^{\tau_n}$ is a $(Q^n,\bG^{X,n})$-local martingale, for every $n\in\N$. Hence, for all $n\in\N$, by \eqref{MRP-1}, there exist two $\bG^{X,n}$-predictable processes $\varphi^n$ and $\psi^n$ such that, for all $t\in[0,T]$:
\be	\label{MRP-2}	\begin{aligned}
L^{\tau_n}_t &= L_0 + \int_0^t\!\varphi^n_u\,\ud\What^{\tau_n}_u
+\int_0^t\!\psi^n_u\,\ud (M^{\bG^X})^{\tau_n}_u
+\int_0^t\!\frac{1}{\Zhat^{\tau_n}_u}\,\ud[\Zhat,L]^{\tau_n}_u	\\
&= L_0 + \int_0^t\!\varphi^n_u\,\ud B^{\tau_n}_u
+\int_0^t\!\psi^n_u\,\ud (M^{\bG^X})^{\tau_n}_u
+\int_0^t\!\frac{1}{\Zhat^{\tau_n}_u}\,\ud[\Zhat,L]^{\tau_n}_u
+\int_0^{t\wedge\tau_n}\!\!\varphi^n_u\,\frac{\mu_u}{V_u}\,\ud u.
\end{aligned} \ee
Since the processes $L^{\tau_n}$, $\int\varphi^n\,\ud B^{\tau_n}$ and $\int\psi^n\,\ud (M^{\bG^X})^{\tau_n}$ are all $(P,\bG^X)$-local martingales, the $\bG^X$-predictable finite variation terms in \eqref{MRP-2} must vanish and 
the representation \eqref{MRP-0} follows by letting the $\bG^X$-predictable processes $\varphi=(\varphi_t)_{0\leq t \leq T}$ and $\psi=(\psi_t)_{0\leq t \leq T}$ be defined as:
$$
\varphi=\sum_{n=1}^{\infty}\indik_{\lsir\tau_{n-1},\tau_n\rsi}\varphi^n
\qquad\text{and}\qquad
\psi=\sum_{n=1}^{\infty}\indik_{\lsir\tau_{n-1},\tau_n\rsi}\psi^n.
$$
\end{proof}

In particular, Proposition \ref{MRP} allows us to obtain an explicit description of the family of all $\bG^X$-martingale deflators. The following corollary is an immediate consequence of Proposition \ref{arb-G^X} and Proposition \ref{MRP}, noting that $[B,M^{\bG^X}]=0$.

\begin{corollary}	\label{NFLVR-G^X}
Suppose that Conditions \ref{B} and \ref{C} hold. Then NFLVR holds for the model \eqref{S}-\eqref{model-2} in the filtration $\bG^X\!$ if and only if NA1 holds and $E\bigl[\cE(\!-\!\int(\mu/V)\ud B+\!\int\psi\,\ud M^{\bG^X}\!)_T\!\bigr]\!\!=\!\!1$ for some $\bG^X$-predictable process $\psi$ such that $\psi\Delta M^{\bG^X}>-1$ and $\int_0^T\!|\psi_t||\ud A^{\bG^X}_t|<\infty$ $P$-a.s.
\end{corollary}

\subsubsection{Stability of no-arbitrage properties with respect to filtration shrinkage}

At this point, one may wonder whether the absence of arbitrage in $\bG$ already implies the absence of arbitrage in the smaller filtration $\bG^X$. 
Intuitively, the answer to such a question is expected to be affirmative, because any outcome of a $\bG^X$-trading strategy should also be realized as the outcome of a $\bG$-trading strategy, since $\bG^X\subset\bG$.
However, one has to prove that stochastic integrals defined in $\bG^X$ can also be viewed as stochastic integrals in the larger filtration $\bG$ (a counterexample can be found in \citeN{Jeulin80}, Theorem 3.23).

\begin{proposition}	\label{stability}
Suppose that Conditions \ref{B} and \ref{C} hold. Then the following assertions hold for the model \eqref{S}-\eqref{model-2}:
\begin{itemize}
\item[(a)]
NA1 in the filtration $\bG$ implies NA1 in the filtration $\bG^X$;
\item[(b)]
NFLVR in the filtration $\bG$ implies NFLVR in the filtration $\bG^X$.
\end{itemize}
\end{proposition}
\begin{proof}
If NA1 holds in $\bG$, then Proposition \ref{arb-G} implies that $\int_0^T\!(\mutilde_t/V_t)^2\,\ud t<\infty$ $P$-a.s. We now show that this implies that $L(S,\bG^X)\subseteq L(S,\bG)$. In view of the discussion preceding the proposition, this will suffice to prove the claim.
Let $h\in L(S,\bG^X)$. Since $S$ is continuous and due to equation \eqref{candec-G^X-1}, this implies that $hS\in L^2_{\text{loc}}\bigl(\int\!V\ud B,\bG^X\bigr)\subseteq L^2_{\text{loc}}\bigl(\int\!V\ud\Wtilde,\bG\bigr)$, i.e., $\int_0^T(h_tS_tV_t)^2\ud t<\infty$ $P$-a.s. Due to the Cauchy-Schwarz inequality:
$$
\int_0^T\!\left|h_tS_t\mutilde_t\right|\ud t
= \int_0^T\Bigl|h_tS_tV_t\frac{\mutilde_t}{V_t}\Bigr|\ud t
\leq \Biggl(\int_0^T\!(h_tS_tV_t)^2\ud t\Biggr)^{\!\!1/2}
\Biggl(\int_0^T\!\left(\frac{\mutilde_t}{V_t}\right)^{\!2}\!\ud t\Biggr)^{\!\!1/2}
\!\!<\infty \qquad \text{$P$-a.s.}
$$
We have thus shown that $hS\in L\bigl(\int_0^{\cdot}\mutilde_u\,\ud u,\bG\bigr)\cap L^2_{\text{loc}}\bigl(\int\!V\ud\Wtilde,\bG\bigr)=L(X,\bG)$, using equation \eqref{candec-G}. Hence, we can conclude that $h\in L(S,\bG)$.
\end{proof}

Conversely, starting from a smaller filtration that satisfies NA1/NFLVR and passing to a larger filtration, it may well happen that arbitrage possibilities are introduced. As a simple example, consider the case where $X=W^1$, so that NFLVR (and, hence, NA1 as well) trivially holds in $\bF^X=\bF^{W^1}$. If $\tau$ is an \emph{honest time} (see \citeN{Jeulin80}, Chapter V) that avoids all $\bF^X$-stopping times, then both NFLVR and NA1 will fail to hold in the enlarged filtrations $\bG=\bG^X$, as shown in \citeN{FontanaJeanblancSong12}.

\subsection{Special cases}	\label{appl}

In this section we analyze two special cases of the general setting described so far, which are important in view of financial applications. 

\subsubsection{Immersion property between $\bF$ and $\bG$}	\label{immersion}

Let us suppose that the filtrations $\bF$ and $\bG$ satisfy the \emph{immersion property} (or \emph{$(\cH)$-hypothesis}, see \citeN{BremaudYor78}) with respect to the random time $\tau$, meaning that all $\bF$-martingales are also $\bG$-martingales. This situation is rather interesting in view of the fact that many random time models considered in financial and insurance mathematics satisfy this property. For instance, many popular credit risk models assume that $\tau$ is a \emph{doubly stochastic} random time (see e.g. \citeN{BieleckiRutkowski02}, Section 8.2). 
In this case, the immersion property holds between $\bF$ and $\bG$ and $\tau$ is a random change point that occurs in an unpredictable way.

The immersion property considerably simplifies the analysis of the model \eqref{S}-\eqref{model-2}. 
Indeed, Condition \ref{C} trivially holds with $\theta^i\equiv0$ for $i=1,2$, and Proposition \ref{arb-G} immediately implies that the model \eqref{S}-\eqref{model-2} satisfies NA1 in $\bG$. However, we cannot a priori exclude the existence of free lunches with vanishing risk.
Furthermore, if $\bF=\bF^{W^1}\vee\bF^{W^2}$, Theorem 2.3 of \citeN{Kusuoka99} provides the martingale representation \eqref{mrp-1} in $\bG$. Hence, part (b) of Proposition \ref{arb-G} together with \eqref{N} yields a complete description of all $\bG$-martingale deflators.

Since $\tau$ is a $\bG^X$-stopping time, it is easy to deduce from \eqref{candec-G}-\eqref{candec-G-3} and Lemma \ref{candec-G^X} that the canonical decomposition of $X$ in the filtration $\bG^X$ coincides with its canonical decomposition in $\bG$. Furthermore, given that $\bar{\theta}\equiv 0$ (and hence $^p\bar{\theta}\equiv 0$), Proposition \ref{arb-G^X} implies that NA1 holds for the model \eqref{S}-\eqref{model-2} considered in the filtration $\bG^X$.

\subsubsection{Stopping times with respect to the filtration $\bF$}	\label{example-stopping}

Let us now consider the case where $\tau$ is a stopping time with respect to $\bF$. For instance, $\tau$ could be defined as the first passage time of one of the two Brownian motions $W^1$ and $W^2$ at some given level. This is the case where $\tau$ is a random change point which is \emph{endogenous} to the model, in the sense that its occurrence is determined by the same stochastic processes which drive the dynamics of $S$.

If $\tau$ is an $\bF$-stopping time, it is evident that $\bG=\bF$. Hence, Condition \ref{C} is trivially satisfied with $\theta^i\equiv0$ for $i=1,2$. In this case, as discussed in Section \ref{immersion}, Proposition \ref{arb-G} implies that NA1 holds in $\bF=\bG$. However, we cannot exclude a priori the existence of free lunches with vanishing risk. Also, when considering the model \eqref{S}-\eqref{model-2} in the filtration $\bG^X$, we are in a situation analogous to that discussed at the end of Section \ref{immersion}.

\section{The price process filtration $\bF^X$}	\label{bF^S}

In this section, we study the model \eqref{S}-\eqref{model-2} with respect to its own filtration $\bF^X$, which is the smallest among all filtrations introduced in Section 2.

\subsection{Canonical decomposition and no-arbitrage properties in $\bF^X$}

In general, the random time $\tau$ is not necessarily an $\bF^X$-stopping time.
Nevertheless, the process $V=(V_t)_{0\leq t \leq T}$ introduced in \eqref{candec-G-2} is $\bF^X$-predictable, as shown in the following lemma.

\begin{lemma}	\label{V-pred}
Suppose that Conditions \ref{B} and \ref{C} hold. Then the process $V=(V_t)_{0\leq t \leq T}$ introduced in \eqref{candec-G-2} is $\bF^X$-predictable.
\end{lemma}
\begin{proof}
For almost every $\omega\in\Omega$ and any $t\in [0,T]$, the quadratic variation $[X]_t=\int_0^tV_u^2\,\ud u$ is
differentiable with respect to $t$ and the derivative is
$$
\frac{\partial}{\partial t}[X]_t(\omega) = V_t^2(\omega) =
\indik_{\{t\leq\tau(\omega)\}}\bigl(\sigma^1\bigl(t,X_t(\omega)\bigr)\bigr)^2
+\indik_{\{t>\tau(\omega)\}}\bigl(\sigma^2\bigl(t,X_t(\omega)\bigr)\bigr)^2.
$$
For all $t\in(0,T]$, the derivative $\frac{\partial}{\partial t}[X]_t$ is $\cF^X_t$-measurable, because it equals the left derivative
$\lim_{\epsilon\searrow 0}\bigl([X]_t-[X]_{t-\epsilon}\bigr)/\epsilon$, while for $t=0$ we have $V^2_0=\bigl(\sigma^1(0,0)\bigr)^2$. This implies that the process $V=(V_t)_{0\leq t \leq T}$ is $\bF^X$-adapted. Being left-continuous, due to part (b) of Condition \ref{B}, it is also $\bF^X$-predictable.
\end{proof}

The next lemma gives the canonical decomposition of $X$ in its own filtration $\bF^X$.

\begin{lemma}	\label{candec-F^X}
Suppose that Conditions \ref{B} and \ref{C} hold. Then the process $X$ admits the following canonical decomposition with respect to the filtration $\bF^X$:
\be	\label{candec-F^X-1}
X_t = \int_0^t\bar{\mu}_u\,\ud u+\int_0^tV_u\,\ud \bar{B}_u,
\qquad \text{for all }t\in[0,T],
\ee
where the $\bF^X$-predictable process $\bar{\mu}=(\bar{\mu}_t)_{0\leq t \leq T}$ is defined as follows, for all $t\in[0,T]$:
\be	\label{candec-F^X-2}	\begin{aligned}
\bar{\mu}_t &:= \,^p(\indik_{\lsi0,\tau\rsi})_t\,\mu^1(t,X_t)+
\,^p\!\left(\indik_{\lsi0,\tau\rsi}\theta^1\right)_t\,\sigma^1(t,X_t)	\\
&\quad
+\,^p(\indik_{\lsir\tau,T\rsi})_t\,\mu^2(t,X_t)
+\,^p\Bigl(\indik_{\lsir\tau,T\rsi}\bigl(\rho\,\theta^1+\sqrt{1-\rho^2}\,\theta^2\bigr)\Bigr)_{\!t}\,\sigma^2(t,X_t),
\end{aligned} \ee
with $^p$ denoting the $\bF^X$-predictable projection and where the process $\bar{B}=(\bar{B}_t)_{0\leq t \leq T}$ is an $\bF^X$-Brownian motion and $V=(V_t)_{0\leq t \leq T}$ is as in \eqref{candec-G-2}.
\end{lemma}
\begin{proof}
The proof is analogous to the proof of Lemma \ref{candec-G^X} noting that the stopping times $\tau_n$ defined therein are also $\bF^X$-stopping times and the process $V$ is $\bF^X$-predictable by Lemma \ref{V-pred}.
\end{proof}

The next proposition answers the question of whether the model \eqref{S}-\eqref{model-2}, considered now with respect to its own filtration $\bF^X$, allows for arbitrage profits. The proof is similar to that of Proposition \ref{arb-G} and, hence, omitted.

\begin{proposition}	\label{arb-F^X}
Suppose that Conditions \ref{B} and \ref{C} hold. Then the following assertions hold for the model \eqref{S}-\eqref{model-2} considered with respect to the filtration $\bF^X$:
\begin{itemize}
\item[(a)]
NA1 holds if and only if $\int_0^T\!(\bar{\mu}_t/V_t)^2\,\ud t$ $<\infty$ $P$-a.s.
\item[(b)]
NFLVR holds if and only if NA1 holds and there exists $N=(N_t)_{0\leq t \leq T}\in\tMloc(\bF^X)$ with $N_0=0$, $\Delta N>-1$ $P$-a.s., $[N,\bar{B}]=0$, such that $E\bigl[\cE\bigl(-\int(\bar{\mu}/V)\ud\bar{B}+N\bigr)_T\bigr]=1$.
\end{itemize}
\end{proposition}

\subsection{The $\bF^X$-martingale representation property}

This section will provide the martingale representation property with respect to
the price filtration $\bF^X$. We shall see that the representation formula and
thus the market completeness or incompleteness
are determined by the relationship of the two volatility functions $\sigma^1$ and
$\sigma^2$. In particular, we provide the representation theorems respectively under Condition
III-(a) and (b) below.

\begin{condition}   \label{E}
(a) Distinct volatility functions
$$\sigma^1(t,x)\neq \sigma^2(t,x) \text{, for any } (t,x)\in[0,T]\times\R.$$
(b) Identical volatility functions
$$\sigma^1(t,x)= \sigma^2(t,x)=:\sigma(t,x) \text{, for all } (t,x)\in[0,T]\times\R.$$
\end{condition}

\subsubsection{Distinct volatility functions}
Let us first analyze the case where $\sigma^1$ and $\sigma^2$ differ everywhere.
Under Condition \ref{E}-(a), the result below explicitly characterizes the filtration $\bF^X$.

\begin{proposition} \label{stopping}
Suppose that Conditions \ref{B}, \ref{C} and \ref{E}-(a) hold. Then $\tau$ is a stopping time for the filtration $\bF^X$ and, consequently, the filtrations $\bF^X$ and $\bG^X$ coincide.
\end{proposition}
\begin{proof}
Condition \ref{E}-(a) and the proof of Lemma \ref{V-pred} together imply that, for all $t\in[0,T]$:
\begin{align}
\label{stopping-1}
\left\{\omega\in\Omega:\tau(\omega)\geq t\right\} & =
\Bigl\{\omega\in\Omega:\frac{\partial}{\partial
t}[X]_t(\omega)=\bigl(\sigma^1\bigl(t,X_t(\omega)\bigr)\bigr)^2\Bigr\}\in\cF^X_t.
\end{align}
Together with the right continuity of the filtration $\bF^X$, this implies that $\{\tau\leq t\}\in\cF^X_t$ for all $t\in[0,T]$, meaning that $\tau$ is an $\bF^X$-stopping time.
\end{proof}

Due to Proposition \ref{stopping}, as long as Condition \ref{E}-(a) holds, all results obtained in Section \ref{bG^S} for the progressively enlarged filtration $\bG^X$ are also true for the filtration $\bF^X$, thus providing a complete description of the filtration $\bF^X$ generated by the process $X$. In particular, Lemma \ref{filtrations} gives an explicit description of the filtration $\bF^X=\bG^X$ as the progressive enlargement of the filtration $\bF^{\What}$ generated by the drifted $\bF^X$-Brownian motion $\What$. Furthermore, any $\bF^X$-local martingale admits the representation obtained in Proposition \ref{MRP}.

\begin{remark}
\label{stopping-2}
 Replacing Condition \ref{E}-(a) with its weaker almost sure version
$$
\sigma^1(t,X_t(\omega))\neq \sigma^2(t,X_t(\omega)) \qquad P\text{-a.s.}, \ \ \text{for any } t \in [0, T],
$$
the result of Proposition \ref{stopping} still remains valid noting that the set
$$
\left\{\omega\in\Omega:\tau(\omega) <t, \sigma^1\bigl(t,X_t(\omega)\bigl) = \sigma^2\bigl(t,X_t(\omega)\bigl)\right\},
$$
appearing in addition on the left-hand side of \eqref{stopping-1}, is a $P$-null set, and hence is in $\cF^X_t$ due to completeness of the filtration $\bF^X$.
\end{remark}

\subsubsection{Identical volatility functions}	\label{ident-vol}

Let us now analyze the case of Condition \ref{E}-(b), where the two
volatility functions $\sigma^1$ and $\sigma^2$ coincide everywhere. Similarly to Remark \ref{stopping-2}, this condition can be relaxed to hold only $P$-a.s. for every $t \in [0, T]$, which by continuity of $X$ and $\sigma^1$ and $\sigma^2$ implies that the processes $(\sigma^1(t, X_t))_{t \in [0, T]}$ and $(\sigma^2(t, X_t))_{t \in [0, T]}$ are indistinguishable. 
Lemma \ref{candec-F^X} gives the following canonical
decomposition of $X$ in its own filtration $\bF^X$:
\be
\label{candec-special} X_t = \int_0^t\!\bar{\mu}_u\,\ud
u+\int_0^t\!\sigma(u,X_u)\,\ud\bar{B}_u,
\qquad \text{for all $t\in[0,T]$,}
\ee
where the process $\bar{\mu}=(\bar{\mu}_t)_{0\leq t
\leq T}$ is as in \eqref{candec-F^X-2}. Since the continuous process
$\sigma(\cdot,X_{\cdot})$ is $\bF^X$-adapted and never attains zero, it is
$\bF^X$-predictable and bounded away from zero. Hence, we can define
a drifted $\bF^X$-Brownian motion $\bar{Y}=(\bar{Y}_t)_{0\leq t
\leq T}$ as follows, for all $t\in[0,T]$:
\be	\label{def-Ybar}
\bar{Y}_t:=\int_0^t\!\frac{1}{\sigma(u,X_u)}\,\ud
X_u=\int_0^t\!\frac{\bar{\mu}_u}{\sigma(u,X_u)}\,\ud u+\bar{B}_t.
\ee
Denote by
$\bF^{\bar{Y}}=(\cF^{\bar{Y}}_t)_{0\leq t \leq T}$ the
right-continuous $P$-augmented natural filtration $\bar{Y}$. We
can easily prove the next lemma, which shows that $\bF^X$ coincides with the filtration $\bF^{\bar{Y}}$ generated by the drifted Brownian motion $\bar{Y}$ (see also Section 3 of \citeN{PhamQuenez01} for related results).

\begin{lemma}   \label{filtrations special}
Suppose that Conditions \ref{B}, \ref{C} and \ref{E}-(b) hold. Then $\bF^X=\bF^{\bar{Y}}$.
\end{lemma}
\begin{proof}
Clearly, the process $\bar{Y}$ defined in \eqref{def-Ybar} is $\bF^X$-adapted and, hence, we have $\bF^{\bar{Y}}\subseteq\bF^X$. To prove the converse inclusion, it suffices to note that the process $X$ can be represented as $X_t = \int_0^t\sigma(u,X_u)\,\ud\bar{Y}_u$ for all $t\in[0,T]$. The same arguments used in the proof of Lemma \ref{filtrations} allow then to show that $\bF^X\subseteq\bF^{\bar{Y}}$, noting that the function $\sigma$ is deterministic.
\end{proof}

As in Section \ref{bG^S}, we can now prove a martingale representation result for the filtration $\bF^X$ in the special case where the two volatility functions $\sigma^1$ and $\sigma^2$ coincide.
We only give an outline of the proof, the arguments being similar to those used for proving Proposition \ref{MRP}. 

\begin{proposition}	\label{MRP special}
Suppose that Conditions \ref{B}, \ref{C} and \ref{E}-(b) hold and assume in addition that NA1 holds in the filtration $\bF^X$. Then every $\bF^X$-local martingale $L=(L_t)_{0\leq t \leq T}$ admits a representation of the form:
$$
L_t = L_0 + \int_0^t\!\varphi_u\,\ud\bar{B}_u,
\qquad\text{ for all }t\in[0,T],
$$
for some $\bF^X$-predictable process $\varphi=(\varphi_t)_{0\leq t \leq T}$ with $\int_0^T\!\varphi^2_t\,\ud t<\infty$ $P$-a.s. 
\end{proposition}
\begin{proof}  
If NA1 holds, then $\Zhat:=\cE(-\int\!\frac{\bar{\mu}}{\sigma(\cdot,X)}\,\ud\bar{B})\in\Mloc(\bF^X)$, with an $\bF^X$-localizing sequence $\{\tau_n\}_{n\in\N}$. For each $n\in\N$, let $dQ^n:=\Zhat_{T\wedge\tau_n}dP$. Girsanov's theorem together with Lemma \ref{filtrations special} then implies that every $(Q^n,\bF^X)$-local martingale (stopped at $\tau_n$) can be represented as a stochastic integral of $\bar{Y}^{\tau_n}$. Using again Girsanov's theorem, the proposition can then be proved by arguing as in the last part of the proof of Proposition \ref{MRP}.
\end{proof}

\begin{remark}
By relying on Proposition \ref{MRP special}, we can prove that the financial market where the asset $S$ is traded with respect to the information contained in its own filtration $\bF^X$ is \emph{complete}, in the sense that for any bounded $\cF^X_T$-measurable non-negative random variable $H$ there exists a couple $(v^H,h^H)\in[0,\infty)\times L(S,\bF^X)$ such that $H=v^H+\int_0^T\!h^H_t\,\ud S_t$ $P$-a.s. 

Note that this result only requires the NA1 condition for the model \eqref{S}-\eqref{model-2} in the filtration $\bF^X$ and does not depend on the validity of NFLVR.
\end{remark}

\section{The initially enlarged filtrations $\bG^{(\tau)}$ and $\bG^{X, (\tau)}$}

In the remainder of the paper, we turn our attention to the model \eqref{S}-\eqref{model-2} considered in the initially enlarged filtrations  $\bG^{(\tau)}$ and $\bG^{X, (\tau)}$. Recall that $\bG^{(\tau)} = (\cG^{(\tau)}_t)_{0 \leq t \leq T}$ is defined by $\cG^{(\tau)}_t:=\bigcap_{s>t}\bigl(\cF_s\vee\sigma(\tau)\bigr)$ and $\bG^{X, (\tau)} = (\cG^{X, (\tau)}_t)_{0 \leq t \leq T}$ by $\cG^{X, (\tau)}_t:=\bigcap_{s>t}\bigl(\cF^X_s\vee\sigma(\tau)\bigr)$, for all $t\in[0,T]$.

\subsection{The initially enlarged filtration $\bG^{(\tau)}$} \label{G^tau}

We begin the analysis of the model \eqref{S}-\eqref{model-2} by assuming that the model  is well-defined in the enlarged filtration $\bG^{(\tau)}$. More precisely, similarly to Condition \ref{C}, the condition imposed below ensures that the driving Brownian motions $W^1$ and $W^2$ remain semimartingales with respect to the filtration $\bG^{(\tau)}$.

\begin{condition} \label{F}
There exist two $\bG^{(\tau)}$-predictable processes $\theta^{1, (\tau)}=(\theta^{1, (\tau)}_t)_{0\leq t \leq T}$ and $\theta^{2, (\tau)}=(\theta^{2, (\tau)}_t)_{0\leq t \leq T}$ and two $\bG^{(\tau)}$-Brownian motions $W^{1,(\tau)}=(W^{1,(\tau)}_t)_{0\leq t  \leq T}$ and $W^{2,(\tau)}=(W^{2,(\tau)}_t)_{0\leq t \leq T}$ such that for all $t\in[0,T]$:
\begin{equation}
\label{Gtau}
W^i_t = W^{i,(\tau)}_t+\int_0^t\!\theta^{i, (\tau)}_u\,\ud u,
\qquad \text{for }i=1,2.
\end{equation}
\end{condition}

\begin{remarks} \label{independence-tau-BM}
\textbf{1)}
We emphasize that here the superscript $^{(\tau)}$ is used to denote processes that are adapted to the filtration $\bG^{(\tau)}$ and should not be confused with the superscript $^\tau$ used earlier for processes stopped at time $\tau$. For example, $Y^{(\tau)} = (Y^{(\tau)}_t)_{0\leq t \leq T}$ is a process adapted to the filtration $\bG^{(\tau)}$ and $Y^{\tau}$ denotes the process $Y$ stopped at $\tau$, i.e. $Y^{\tau} = (Y_{t \wedge \tau})_{0\leq t \leq T}$.

\textbf{2)}
Note that the $\bG^{(\tau)}$-Brownian motions $W^{i,(\tau)}$, for $i=1, 2$, are independent of the random time $\tau$, due to the independence of Brownian increments together with $\sigma(\tau)\subseteq\cG^{(\tau)}_0$.
\end{remarks}

Condition \ref{F}, as well as Condition \ref{C}, are satisfied under the classical \emph{density hypothesis} due to \citeN{Jacod85}, which is typically used in the literature when dealing with initial enlargements of filtrations and which assumes that the (regular) $\cF_t$-conditional law of $\tau$ admits a density with respect to the unconditional law of $\tau$, for all $t\in[0,T]$.
In particular, if $\tau$ is independent of $\bF$, the density hypothesis is trivially satisfied (with the constant density 1). In the latter case, condition \eqref{Gtau} obviously holds with $\theta^{i, (\tau)}\equiv 0$, for $i=1, 2$.

Under Condition \ref{F}, the process $X=(X_t)_{0\leq t  \leq T}$ admits the following canonical decomposition with respect to the filtration $\bG^{(\tau)}$ (compare with \eqref{candec-G}):
 \be	\label{candec-Gtau}
X_t = \int_0^t\!\tilde{\mu}^{(\tau)}_u\,\ud u+\int_0^t\!V_u\,\ud W^{(\tau)}_u,
\qquad \text{for all }t\in[0,T],
\ee
where $\tilde{\mu}^{(\tau)}=(\tilde{\mu}^{(\tau)}_t)_{0\leq t \leq T}$ and $W^{(\tau)}=(W^{(\tau)}_t)_{0\leq t \leq T}$ are defined as follows, for all $t\in[0,T]$:
\begin{gather}
\tilde{\mu}^{(\tau)}_t :=
\indik_{\{t\leq\tau\}}\left(\mu^1(t,X_t)+\sigma^1(t,X_t)\,\theta^{1, (\tau)}_t\right) \nonumber \\
+\indik_{\{t>\tau\}}\left(\mu^2(t,X_t)+\sigma^2(t,X_t)\bigl(\rho\,\theta^{1,(\tau)}_t+\sqrt{1-\rho^2}\,\theta^{2,(\tau)}_t\bigr)\right)	 \label{candec-Gtau-1} \\
W^{(\tau)}_t := 
W^{1,(\tau)}_{t\wedge\tau}+\rho\,\bigl(W^{1,(\tau)}_{t\vee\tau}-W^{1,(\tau)}_{\tau}\bigr)
+\sqrt{1-\rho^2}\,\bigl(W^{2,(\tau)}_{t\vee\tau}-W^{2,(\tau)}_{\tau}\bigr)
\label{candec-Gtau-3}
\end{gather}
and $V=(V_t)_{0\leq t \leq T}$ is defined in \eqref{candec-G-2}. By computing its quadratic variation one can easily verify that the process $W^{(\tau)}$ is a $\bG^{(\tau)}$-Brownian motion. 
Now we are ready to establish the no-arbitrage properties of the model \eqref{S}-\eqref{model-2} considered in the initially enlarged filtration $\bG^{(\tau)}$. The proof of the proposition below relies on the same arguments as the proof of Proposition \ref{arb-G}.
We define the $\bG^{(\tau)}$-predictable process $\theta^{(\tau)}=(\theta^{(\tau)}_t)_{0\leq t \leq T}$, for all $t\in[0,T]$:
\be	\label{thetatau}
\theta^{(\tau)} := \indik_{\lsi0,\tau\rsi}\theta^{1, (\tau)}
+\indik_{\lsir\tau,T\rsi}\bigl(\rho\,\theta^{1, (\tau)}+\sqrt{1-\rho^2}\,\theta^{2, (\tau)}\bigr).
\ee

\begin{proposition}	\label{arb-Gtau}
Suppose that Conditions \ref{B} and \ref{F} hold. Then the following assertions hold for the model \eqref{S}-\eqref{model-2} considered with respect to the filtration $\bG^{\tau}$:
\begin{itemize}
\item[(a)]
NA1 holds if and only if  $\int_0^T\!(\theta^{(\tau)}_t)^2\,\ud t<\infty$ $P$-a.s., with the latter condition being equivalent to $\int_0^T\!(\tilde{\mu}^{(\tau)}_t/V_t)^2\,\ud t<\infty$ $P$-a.s.;
\item[(b)]
NFLVR holds if and only if NA1 holds and there exists $N^{(\tau)}=(N^{(\tau)}_t)_{0\leq t \leq T}\in\tMloc(\bG^{(\tau)})$ with $N^{(\tau)}_0=0$, $\Delta N^{(\tau)} > -1$ $P$-a.s. and $[N^{(\tau)}, W^{(\tau)}]=0$ satisfying $E[\mathcal{E}(-\int(\tilde{\mu}^{(\tau)}/V)\,\ud W^{(\tau)}+N^{(\tau)})_T]=1$.
\end{itemize}
\end{proposition}

\begin{remark}[\emph{The $\bG^{(\tau)}$-martingale representation property}]
Suppose that $\bF=\bF^{W^1} \vee \bF^{W^2}$ and that the density hypothesis holds with a $P$-a.s. strictly positive density. 
Then, Proposition 5.3(i) in \citeN{CallegaroJeanblancZargari11} implies that 
the pair $(W^{1,(\tau)},W^{2,(\tau)})$ enjoys the martingale representation property in the filtration $\bG^{(\tau)}$.
Hence, we can obtain an explicit representation of
a $\bG^{(\tau)}$-local martingale $N^{(\tau)}=(N^{(\tau)}_t)_{0\leq t \leq T}$ appearing in Proposition \ref{arb-Gtau}(b):
\begin{eqnarray}	
\label{Ltau}
N^{(\tau)}_t & = & \int_0^{t\wedge\tau}\!\varphi^{2, (\tau)}_u\,\ud W^{2, (\tau)}_u
+\indik_{\{\rho=0\}}\!\!\int_{\tau}^{t\vee\tau}\!\varphi^{1, (\tau)}_u\,\ud W^{1, (\tau)}_u \\
\nonumber && \quad +\indik_{\{\rho\neq0\}}\!\Biggl(\int_{\tau}^{t\vee\tau}\!\varphi^{3, (\tau)}_u\,\ud W^{2, (\tau)}_u
-\frac{\sqrt{1-\rho^2}}{\rho}\int_{\tau}^{t\vee\tau}\!\varphi^{3, (\tau)}_u\,\ud W^{1, (\tau)}_u\Biggr),
\end{eqnarray}
where $\varphi^{j,(\tau)}= (\varphi^{j,(\tau)}_t)_{0\leq t\leq T}$ is a $\bG^{(\tau)}$-predictable process satisfying the integrability condition $\int_0^T\!(\varphi^{j,(\tau)}_t)^2 \,\dt < \infty$ $P$-a.s., for $j=1,2,3$. Together with Proposition \ref{arb-Gtau}, this gives a full characterization of the set of all ELMMs for the model \eqref{S}-\eqref{model-2} in the filtration $\bG^{(\tau)}$.
\end{remark}

\subsection{The initially enlarged filtration $\bG^{X, (\tau)}$}

In this section, we study the model \eqref{S}-\eqref{model-2} in the filtration $\bG^{X, (\tau)}$. Here we might be concerned with a non-standard initial enlargement since we are enlarging initially the filtration $\bF^X$ with respect to which the random time $\tau$ can already be a stopping time (see for instance Proposition \ref{stopping}). Thus, the density hypothesis cannot be imposed. In this case a method developed in \citeN[Chapter 12]{Yor97} can be applied. An example of such an initial enlargement is treated on page 53 of \citeN{Jeulin80} and in \citeN{JeanblancLeniec13}.

Let us begin by stating the canonical decomposition of $X$ with respect to the filtration $\bG^{X, (\tau)}$. It is obtained by projecting the canonical decomposition with respect to $\bG^{(\tau)}$ onto $\bG^{X, (\tau)}$ and the proof is analogous to that of Lemma \ref{candec-G^X}.

\begin{lemma}
\label{candec-F^X,tau}
Suppose that Conditions \ref{B} and \ref{F} hold. Then the process $X$ admits the following canonical decomposition with respect to the filtration $\bG^{X, (\tau)}$:
\be	\label{candec-F^X,tau-1}
X_t = \int_0^t\!\mu^{(\tau)}_u\,\ud u+\int_0^t\!V_u\,\ud B^{(\tau)}_u,
\qquad \text{for all }t\in[0,T],
\ee
where the $\bG^{X, (\tau)}$-predictable process $\mu^{(\tau)}=(\mu^{(\tau)}_t)_{0\leq t \leq T}$ is defined as follows, for all $t\in[0,T]$:
\begin{eqnarray}	
\label{candec-F^X,tau-2}
\mu^{(\tau)}_t & := & \indik_{\{t\leq\tau\}}\left(\mu^1(t,X_t)+\sigma^1(t,X_t)\,\,^p\theta^{1, (\tau)}_t\right) \\
\nonumber && \quad +\indik_{\{t>\tau\}}\left(\mu^2(t,X_t)
+\sigma^2(t,X_t)\bigl(\rho\,\,^p\theta^{1, (\tau)}_t+\sqrt{1-\rho^2}\,\,^p\theta^{2, (\tau)}_t\bigr)\right),
\end{eqnarray}
with $^p\theta^{i, (\tau)}$ denoting the $\bG^{X, (\tau)}$-predictable projection of $\theta^{i, (\tau)}$, for $i=1,2$, and where the process $B^{(\tau)}=(B^{(\tau)}_t)_{0\leq t \leq T}$ is a $\bG^{X, (\tau)}$-Brownian motion (independent of $\tau$, since $\sigma(\tau)\subseteq\cG^{X,(\tau)}_0$) and $V=(V_t)_{0\leq t \leq T}$ is defined as in \eqref{candec-G-2}.
\end{lemma}

Let us now check if the model \eqref{S}-\eqref{model-2}, considered with respect to the filtration $\bG^{X, (\tau)}$, admits arbitrages of the first kind. We denote by $^p\theta^{(\tau)}$ the $\bG^{X, (\tau)}$-predictable projection of the process $\theta^{(\tau)}$ introduced in \eqref{thetatau}:
 \be	
 \label{predictable-thetatau}
^p\theta^{(\tau)} := \indik_{\lsi0,\tau\rsi}\, ^p\theta^{1, (\tau)}
+\indik_{\lsir\tau,T\rsi}\bigl(\rho\,\, ^p\theta^{1, (\tau)}+\sqrt{1-\rho^2}\,\, ^p\theta^{2, (\tau)}\bigr).
\ee
The next proposition characterizes the validity of NA1 with respect to the filtration $\bG^{X,(\tau)}$. The proof is analogous to that of Proposition \ref{arb-G}(a) and, hence, omitted.

\begin{proposition}	\label{arb-F^X, tau}
Suppose that Conditions \ref{B} and \ref{F} hold. Then the model \eqref{S}-\eqref{model-2} considered in the filtration $\bG^{X, (\tau)}$ satisfies NA1 if and only if $\int_0^T\!(^p\theta^{(\tau)}_t)^2\,\ud t<\infty$ $P$-a.s., with the latter condition being equivalent to $\int_0^T\!(\mu^{(\tau)}_t/V_t)^2\,\ud t<\infty$ $P$-a.s.
\end{proposition}

Similarly as in Sections \ref{MRP-G^X} and \ref{ident-vol}, 
let us define the $\bG^{X, (\tau)}$-adapted continuous process $\widehat{Y}^{(\tau)} = (\widehat{Y}^{(\tau)}_t)_{0\leq t\leq T}$ as follows, for all $t\in[0,T]$:
$$
\widehat{Y}^{(\tau)}_t := \int_0^t\frac{\mu^{(\tau)}_u}{V_u}\,\du +  B^{(\tau)}_t,
$$
where $\mu^{(\tau)}$ and $B^{(\tau)}$ are as in Lemma \ref{candec-F^X,tau}. Let us denote by $\bF^{\What^{(\tau)}}=(\cF^{\What^{(\tau)}}_t)_{0\leq t \leq T}$ the right-continuous $P$-augmented natural filtration of $\What^{(\tau)}$. By relying on the same arguments of Lemma \ref{filtrations}, we can characterize the filtration $\bG^{X, (\tau)}$ as an initial enlargement of the filtration $\bF^{\What^{(\tau)}}$ with respect to $\tau$.

\begin{lemma}
\label{filtrations-tau}
Suppose that Conditions \ref{B} and \ref{F} hold. Then $\bG^{X, (\tau)} = \bG^{\widehat{Y}^{(\tau)}\!,(\tau)}$, where  $\bG^{\widehat{Y}^{(\tau)}\!,(\tau)}$ is the initial enlargement of $ \bF^{\widehat{Y}^{(\tau)}}$ with respect to $\tau$, i.e., $ \cG^{\widehat{Y}^{(\tau)}\!,(\tau)}_t := \bigcap_{s>t}\bigl( \cF^{\widehat{Y}^{(\tau)}}_s \vee \sigma(\tau)\bigr)$, for all $t\in[0,T]$.
\end{lemma}

By relying on Lemma \ref{filtrations-tau} we are able to study the filtration  $\bG^{X, (\tau)}$ using standard techniques for initial enlargements of filtrations.
We begin by proving a martingale representation property for the filtration $\bG^{X, (\tau)}$. The proof relies on a localization procedure similar to the one in the proof of Proposition \ref{MRP}. To obtain the martingale representation property for the initial enlargements constructed in the proof one uses the fact that any filtration initially enlarged with an independent random time is immersed in that enlargement. The independence of $\tau$ follows similarly as in part 2 of Remarks \ref{independence-tau-BM}.

\begin{proposition}
\label{MRP-F^X,tau}
Suppose that Conditions \ref{B} and \ref{F} hold and assume that NA1 holds in the filtration $\bG^{X, (\tau)}$. Then every $\bG^{X, (\tau)}$-local martingale $L^{(\tau)}=(L^{(\tau)}_t)_{0\leq t \leq T}$ admits a representation of the form:
\be	\label{MRP-F^X,tau-0}
L^{(\tau)}_t = L^{(\tau)}_0+\int_0^t\!\varphi^{(\tau)}_u\,\ud B^{(\tau)}_u,
\qquad\text{for all }t\in[0,T],
\ee
for some $\bG^{X, (\tau)}$-predictable process $\varphi^{(\tau)}=(\varphi^{(\tau)}_t)_{0\leq t \leq T}$ with $\int_0^T\!(\varphi^{(\tau)}_t)^2\,\ud t<\infty$ $P$-a.s.
\end{proposition}

Lemma \ref{candec-F^X,tau} and Proposition \ref{MRP-F^X,tau} allow us to characterize the validity of NFLVR in $\bG^{X,(\tau)}$. Similarly to Proposition \ref{arb-G} and Proposition \ref{stability}, we obtain:

\begin{proposition}
\label{NFLVR-F^X-tau}
Suppose that Conditions \ref{B} and \ref{F} hold. Then:
\begin{itemize} 
\item[(a)] NFLVR holds in the filtration $\bG^{X, (\tau)}$ if and only if NA1 holds and $E[\cE(-\int(\mu^{(\tau)}/V)\,\ud B^{(\tau)})_T]=1$.
\item[(b)] If the model \eqref{S}-\eqref{model-2} satisfies NA1 (resp. NFLVR) in the filtration $\bG^{(\tau)}$, then NA1 (resp. NFLVR) holds in the filtration $\bG^{X, (\tau)}$ as well.
\end{itemize}
\end{proposition}

\begin{remark}
As a consequence of Proposition \ref{MRP-F^X,tau}, the financial market $(S, \bG^{X, (\tau)} )$ is complete (up to a random initial endowment). More precisely, we have: any bounded $\cG^{X, (\tau)}_T$-measu\-rable non-negative random variable $H$ admits the representation $H= f^H(\tau) + \int_0^T\!h_t^H \,\ud S_t$ $P$-a.s., for some Borel-measurable function $f^H:\R_+\rightarrow\R_+$ and for $h^H\in L(S,\bG^{X,(\tau)})$.
\end{remark}

\begin{remark}[\emph{Connection between the Brownian motions $B$ and $B^{(\tau)}$}]
Making use of the canonical decompositions of the process $X$ in the progressively enlarged filtration $\bG^X$ and the initially enlarged filtration $\bG^{X,(\tau)}$ given in Lemma \ref{candec-G^X} and Lemma \ref{candec-F^X,tau}, respectively, we obtain a link between the $\bG^{X}$-Brownian motion $B$ and  the $\bG^{X, (\tau)}$-Brownian motion $B^{(\tau)}$:
\begin{equation}
\label{connection-BM-G^X-G^X,tau}
B_t = B^{(\tau)}_t + \int_0^t\!\frac{1}{V_s} \left( ^p \tilde{\mu}^{(\tau)}_s - \, ^p \tilde{\mu}_s \right) \ds, \qquad \text{for all } t \in [0, T].
\end{equation}
Here, $^p \tilde{\mu}^{(\tau)}$ is the $\bG^{X, (\tau)}$-predictable projection of the process $\tilde{\mu}^{(\tau)}$  defined in \eqref{candec-Gtau-1} and $^p \tilde{\mu}_s$ is the $\bG^{X}$-predictable projection of the process $\tilde{\mu}$ defined in \eqref{candec-G-1}. Note that the process $\frac{1}{V} \left( ^p \tilde{\mu}^{(\tau)} - \, ^p \tilde{\mu}\right)$  is $\bG^{X, (\tau)}$-predictable since $V$ and $^p \tilde{\mu}^{(\tau)}$ are $\bG^{X, (\tau)}$-predictable and $^p \tilde{\mu}$ is $\bG^{X}$-predictable (recall that $\bG^{X} \subset \bG^{X, (\tau)}$).
\end{remark}

\section{Conclusion and further developments}

In this paper we have studied a class of asset price models with a change point, imposing only minimal assumptions on the random time and on the driving Brownian motions. We characterize the model by its properties of martingale representation, completeness or incompleteness and two notions of arbitrage. The analysis of the model is undertaken in all the filtrations that naturally arise in this framework.

The model can be generalized to incorporate multiple change points by enlargement of filtrations with a sequence of random times and to incorporate discontinuous semimartingales as the driving processes by techniques for jump processes. If a martingale representation property holds in a filtration generated by the two underlying semimartingales (this is true e.g. for L\'evy processes or marked point processes), then martingale representation theorems can also be obtained for various filtrations related to the model. Similarly, our results can also be generalized to the case where the coefficients $\mu$ and $\sigma$ in \eqref{model-2} are only $\bF$-progressively measurable processes. The extension of other results from this paper, in particular the characterization of the no-arbitrage properties, depends on the specific class of semimartingales which are used as driving processes.

Besides the obvious application to the modeling of financial asset prices, the results of the present paper can also be of interest in view of credit risk, interest rate modeling and energy markets, where sudden changes in the dynamics of the underlying processes are naturally observed.
A further line of research concerns stochastic control problems under incomplete information, either in the manner of utility maximization like in, for example, \citeN{BjorkDavisLanden10}, \citeN{Lakner95}, 
and \citeN{PhamQuenez01},  or extending a method lately introduced by \citeN{KaratzasLi11}.

\bibliographystyle{chicago}
\bibliography{references_ChangePoint}

\end{document}